\pdfoutput=1
\documentclass{article}
\usepackage{spconf,amsmath,graphicx,hyperref,amssymb,amsthm,mathtools,cleveref}

\DeclareMathOperator*{\diag}{diag} % diagonal matrix
\newcommand{\mA}{\mathbf{A}} % matrix
\newcommand{\mQ}{\mathbf{Q}} % matrix
\newcommand{\mI}{\mathbf{I}} % matrix
\newcommand{\mF}{\mathbf{F}} % matrix
\newcommand{\mD}{\mathbf{D}} % matrix
\newcommand{\va}{\mathbf{a}} % vector
\newcommand{\vb}{\mathbf{b}} % vector
\newcommand{\vd}{\mathbf{d}} % vector
\newcommand{\vy}{\mathbf{y}} % vector
\newcommand{\vz}{\mathbf{z}} % vector
\newcommand{\vf}{\mathbf{f}} % vector
\newcommand{\vone}{\mathbf{1}} % vector of ones
\DeclareMathOperator*{\E}{\mathbb{E}} % expectation
\DeclareMathOperator*{\Cov}{\mathbf{Cov}} % expectation
\newtheorem{lemma}{Lemma}
\newtheorem{propn}{Proposition}
\newtheorem{thm}{Theorem}
\DeclarePairedDelimiter\ceil{\lceil}{\rceil}
\DeclarePairedDelimiter\floor{\lfloor}{\rfloor}

\title{Near-Optimal Coded Apertures for Imaging via Nazarov's Theorem}
\name{Ganesh Ajjanagadde$^{\star}$ \qquad Christos Thrampoulidis$^{\dagger}$
    \qquad Adam Yedidia$^{\star}$ \qquad Gregory Wornell$^{\star}$
\thanks{This work was supported, in part, by DARPA under
Contract No. HR0011-16-C-0030, and by NSF under Grant No. CCF-1816209.}}
\address{${\star}$Department of Electrical Engineering and Computer Science,
Massachusetts Institute of Technology \\
${\dagger}$Department of Electrical and Computer Engineering, University of California, Santa Barbara}
\begin{document}
\ninept
\maketitle
\begin{abstract}
We characterize the fundamental limits of coded aperture imaging systems up to
universal constants by drawing upon a theorem of Nazarov regarding Fourier transforms.
Our work is performed under a simple propagation and sensor model that
accounts for thermal and shot noise, scene correlation, and exposure time.
Focusing on mean square error as a measure of linear reconstruction quality,
we show that appropriate application of a theorem of Nazarov leads to essentially
optimal coded apertures, up to a constant multiplicative factor in exposure time.
Additionally, we develop a heuristically efficient algorithm
to generate such patterns that explicitly takes into account scene correlations.
This algorithm finds apertures that correspond to local optima of a certain
potential on the hypercube, yet are guaranteed to be tight.
Finally, for i.i.d.\ scenes, we show improvements upon prior work by using
spectrally flat sequences with bias. The development focuses on 1D apertures for conceptual
clarity; the natural generalizations to 2D are also discussed.
\end{abstract}

\section{Introduction}
\label{sec:intro}

Certain modern imaging systems, especially those operating at high
frequencies, use coded apertures. In these systems, a spatial mask that
selectively blocks light from reaching the sensor is used as opposed to a
traditional lens. The scene is then recovered by suitable post-processing.
Perhaps the earliest and simplest instance of coded aperture imaging is the
pinhole structure; see, e.g., \cite{young1971} for a survey. The development of
X-ray and gamma-ray astronomy gave rise to more sophisticated coded
apertures~\cite{ables1968,dicke1968} to get around the lack of lenses and
mirrors in such settings. Both proposed using random
blockage patterns with a specified mean transmittance as a method to increase the
aperture size as compared to the classical pinhole while retaining its resolution
benefits.

More modern developments include the usage of uniformly redundant arrays (URA)
to improve upon random on-off patterns~\cite{fenimore1978}, anti-pinhole imaging~\cite{cohen1982},
as well as the combining of mask and lens in order to, e.g., facilitate depth
estimation~\cite{levin2007}, deblur out-of-focus elements in an
image~\cite{zhou2011}, enable motion deblurring~\cite{raskar2006}, and/or
recover 4D lightfields~\cite{veeraraghavan2007}. Even more recent work seeks
to forgo lenses altogether to decrease costs and meet physical
constraints~\cite{duarte2008,asif2017}. Understanding coded apertures is also
relevant in non-line-of-sight applications where masks naturally occur as
scene occlusions~\cite{torralba2012,thrampoulidis2018}.

In light of the increased importance of coded apertures, prior work~\cite{yedidia2018}
described a model under which they can be analyzed. This model
uses far-field geometric optics to model light propagation and a sensor model
that includes thermal and shot noise components. Together with mutual
information (MI) as a performance metric,~\cite{yedidia2018} compared the classical
random on-off apertures~\cite{ables1968,dicke1968} of varying intensity to the
``spectrally flat'' patterns with transmissivity $1/2$ (same as the URA
of~\cite{fenimore1978}). Among other things, the analysis showed that when
shot noise dominates thermal noise, randomly generated masks with lower
transmissivity than $1/2$ offered greater performance compared to spectrally flat patterns
of transmissivity $1/2$.

This paper extends the work of~\cite{yedidia2018} in multiple respects that may be broadly
grouped into the following three main contributions.

First, we refine the model of~\cite{yedidia2018} by incorporating exposure time.
Here, we analyze linearly-constrained minimum mean square error (LMMSE)
estimatation as opposed to MI given its direct operational relevance, though we
remark in advance that our conclusions carry over to the MI criterion used in~\cite{yedidia2018}; see
Sec.~\ref{sec:conclusion}.

Second, we remark upon the existence and construction of spectrally flat sequences
with transmissivities $1/8, 1/4$ in addition to $1/2$. This
extends the range of parameters where we have a sharp characterization of
optimal coded apertures in our framework, and gives a tight answer to the
problem of optimal coded apertures for i.i.d.\ scenes; see Props.~\ref{propn:spec_flat_simple},
~\ref{propn:spec_flat_comp} for precise statements.

Third, we provide optimal (up to a universal constant) coded apertures, both
in 1D as well as in 2D, applicable for any prior on the spectrum of the scene at hand.
The sense of tightness of the optimality is given precisely in Prop.~\ref{propn:nazarov_power}.
This includes (but is not limited to) the naturally occuring power law~\cite{millane2003}($f^{-\gamma}$-prior).
Our aperture design naturally varies depending on the choice of prior, and we
provide a (heuristically) efficient greedy algorithm for their generation.
Essentially all the required mathematical results stem from a beautiful theorem
of Nazarov~\cite[p. 5]{nazarov1997}
combined with classical waterfilling for spectrum allocation. We note
that~\cite[pp. 9-11]{boche2016} has identified other applied problems for which
Nazarov's theorem provides conceptual clarity and/or solutions.

\section{Model}

We first describe our model, and discuss how it differs from that in
\cite{yedidia2018}.  We use the standard Poisson model of classical
optics for photon counting, and emphasize its dependence on the
exposure time $t$.   The analysis of MI under Poisson models is
cumbersome, and even with mean square error (MSE) it is often unclear
how to achieve optimal MSE in practice.  As such, the standard
estimation process is linear; indeed, the work
of~\cite{ables1968,dicke1968} used correlation decoders. In fact, both
~\cite{ables1968,dicke1968} give beautiful analog realizations of such
a decoder. Accordingly, we emphasize LMMSE.  We note that if one used a
Gaussian model instead, LMMSE is the same as MSE, and MSE is in turn
essentially equivalent to MI in the low SNR
limit~\cite{stam1959,guo2005}.  LMMSE depends purely on first and
second moments, so in our mathematical study we do not emphasize the
specific Poisson statistics.

We use a 1D model, as in~\cite{yedidia2018}, to simplify the
exposition of the concepts and results.  We emphasize that all of the
results of this paper generalize naturally to the analogous 2D model,
whose discussion we defer to Sec.~\ref{sec:conclusion}.

Let $\vf$ denote the intensities of the unknown 1D scene of length $n$
of expected total power $J$. Let $\E[\vf]=(J/n)\vone,
\Cov[\vf]=\mQ$. We assume $\mQ$ is circulant and diagonalized as $\mQ
= \mF_n^* \mD \mF_n$; $\mF_n$ is the unitary discrete Fourier
transform (DFT) matrix and $\mD =
\diag(\vd)$. The measurements at the imaging plane are denoted $\vy_j,
j \in [n]$ and the $n \times n$ transfer matrix $\mA$ models the
aperture.  We assume its entries all satisfy $0 \le \mA_{ji} \le 1/n$
to model that the light can not be redirected, and $\sum_{j} \mA_{ji}
\le 1$ to model local conservation of power.  An ideal, perfectly
focused, lens may be treated in this setup by $\mA = \mI$, as it
redirects light perfectly.

We assume $\mA_{ji} = (1/n)\va_{i-j \pmod n}$ for a $\va \ge 0$, i.e. $\mA$ is
circulant.
Let $\rho(\va) = (1/n) \sum_i \va_i$ be the~\emph{transmissivity} of the aperture.
The noise component is denoted by $\vz$ and its statistics are given by
$\E[\vz] = 0, \Cov[\vz] = (t(W+J\rho)/n)\mI$, where $W,J$ correspond to thermal
and shot noise respectively, and $t$ is the \emph{exposure time}.
With these, our measurement model is then given by $\vy_j = t\sum_i \mA_{ji}\vf_i + \vz_j$,
which leads to the following expression for the LMMSE of estimating $\vf$ from $\vy$:
\begin{equation}
\label{eqn:lmmse}
m(n,t,W,J,\vd,\va) = \sum_{i=0}^{n-1}
\frac{1}{\frac{1}{\vd_i}+\frac{t|\hat{\va}_i|^2}{n(W+J\rho(\va))}}.
\end{equation}
Here, $\hat{\va}$ is the DFT of $\va$. In general, we
assume $\vd_i = (1/n)d(i/n)$ are $n$ equally spaced samples from a
nonnegative, bounded, continuous function $d(x)$ on $[0,1]$ with symmetry $d(x)=d(1-x)$
and normalized so that $d(0)=\theta$. For example, i.i.d.\ scenes correspond to
$d(x)=\theta$. We note that our main result, Prop.~\ref{propn:nazarov_power},
holds in greater generality. The above restriction on the form of $\vd$ simply
ensures correct physical scaling (invariant with respect to $n$) of the variance
of total scene intensity coming from an arbitrary direction.

It is instructive to compare an ideal lens to a
mask with respect to \eqref{eqn:lmmse}, as a function of exposure
time. An ideal lens satisfies $\mA = \mI$, (i.e., $\va =
(n,0,\dots,0)$). Thus $\hat{\va} = (n,n,\dots,n)$. Then
from~\eqref{eqn:lmmse}, it can be readily seen that for a $t$ growing
with $n$ (say $t = \log(n)$), the LMMSE decays to $0$ as $n
\rightarrow \infty$.  On the other hand, the entry-wise restriction
$\va \in [0,1]$ that holds for a mask results in a significant
reduction in $\|\hat{\va}\|_2$. Due to this, in order to get an LMMSE
that is bounded away from the trivial $\int d(x)\,\mathrm{d}x$, one
needs an exposure time that is $\Omega(n)$. Of course, this is not
surprising; there are strong benefits to lenses when they are
available. The need for long exposure times for coded apertures is
also a known phenomenon, consistent with the emphasis
of~\cite{dicke1968} on ``hypothesis tests'' between scenes as opposed
to resolving full detail.

One way to interpret increased $t$ is that it reduces noise relative to the signal.
All our main results established in the sequel
(~\cref{eqn:spec_flat_simple,eqn:spec_flat_comp,eqn:nazarov_power}) show that
one can construct apertures that are guaranteed to be tight within a constant
factor of $t$. Under the above interpretation, what we
establish rigorously is that our results are tight to within a universal constant
number ($\approx 18.30$) of dB, regardless of the scene correlation structure given by $\vd$.
This factor may be read off from $2M(n)^2$ of Prop.~\ref{propn:nazarov_power}.

\section{Results}
\label{sec:results}

The goal of optimal aperture design (aka optimal $\va$) is to minimize the LMMSE
formula subject to the scene model, denoted as follows:
\[
    m^*(n,t,W,J,\vd) \triangleq \min_{\va} m(n,t,W,J,\vd,\va).
\]

Let us first understand why the minimization above is a challenging problem.
Consider the even simpler problem in which the optimal transmissivity, say $\rho_0$,
is given to us. Then, although $\va \in [0,1], \rho(\va) = \rho_0$ is a convex
constraint, the LMMSE~\eqref{eqn:lmmse} which we wish to minimize is neither
convex nor quasiconvex in $\va$, since $1/(1+cx^2)$ lacks
any of these behaviors.

In order to solve this problem, our general approach is as follows. First, we
use Parseval's identity that relates time and frequency space. Under a fixed power
budget, it is easy to solve for the optimal spectrum allocation $|\hat{\va}_i|^2$
by studying the well-behaved and convex $1/(1+cx)$ that has a solution given by waterfilling~\eqref{eqn:lmmse_lb}.
Next, we are faced with the ``coefficient problem'' of finding a $\va \in [0,1]$
with given spectrum allocation. To address this, we appropriately apply a
theorem of Nazarov~\cite[p. 5]{nazarov1997}. An exposition of Nazarov's work
together with the context he draws from (e.g., the geometric ideas of~\cite{bang1951},
along with the analytic ideas of~\cite{leeuw1977}) may be found in~\cite{ball2001}.

\subsection{Lower bound}

We first derive a lower bound for LMMSE~\eqref{eqn:lmmse} based on
waterfilling (see, e.g., \cite[Thm 19.7]{yury2018}). For notational ease, we let $\gamma = t/(n(W+J\rho))$
throughout.
\begin{propn}
\label{propn:lmmse_lb}
Let $\va$ satisfy $\rho(\va)=\rho$. Then:
\begin{equation}
\label{eqn:lmmse_lb}
m(n,t,W,J,\vd,\va) \ge \frac{1}{\frac{n}{\theta}+\gamma n^2\rho^2} +
\sum_{i=1}^{n-1} \frac{1}{\frac{1}{\vd_i}+\gamma P_i}.
\end{equation}
Here $P_i = (1/\gamma)(T-1/d_i)^+$ and total power $P = \sum_{i=1}^{n-1} P_i = n(\floor{n\rho}+(n\rho-\floor{n\rho})^2)-n^2\rho^2$.
Also note $P \le n^2\rho(1-\rho).$
We remark that~\eqref{eqn:lmmse_lb} is sharp if and only if $|\hat{\va}_i|^2 = P_i$ for $0 < i < n$.
\end{propn}

\begin{proof}
We have $\hat{\va}_0 = n\rho$, giving the first term. For the nonzero
frequencies, we use the fact that the maximum of $\sum_i x_i^2$ subjected to $x_i \in
[0,1]$ and $\sum x_i=r$ is $\floor{r}+(r-\floor{r})^2$.
This, together with Parseval's identity, yields an upper bound on the power of
the nonzero frequencies. Waterfilling, modified to study
$\frac{-1}{1+ax}$ as opposed to $\log(1+ax)$, then gives the proposition.
The floors are removed to get the $P$ upper bound by $x^2 \le x$ for $0 \le x \le 1$.
\end{proof}

Note that minimizing the right hand side over $\rho$ gives a \emph{lower bound}
on $m^*(n,t,W,J,\vd)$. This task is trivial numerically, but in general
difficult analytically. We denote this optimal $\rho$ by $\rho^*$ henceforth.
\subsection{Upper bound}

Our goal here has been set from~\eqref{eqn:lmmse_lb}. Conceptually, the design
issue is finding a $\va \in [0,1]$ with prescribed lower bounds $|\hat{\va}_i|^2 \ge P_i$.
In general, this is impossible to do, and thus our lower bound~\eqref{eqn:lmmse_lb} is not sharp in all settings.
However, it should be noted that sharp cases do exist. Perhaps the
conceptually simplest example is the analog of~\eqref{eqn:lmmse_lb} for a
lens, where our bound is sharp.

Our general approach is to simply step back by a factor $C$ and obtain a $\va
\in [0,1]$ with $|\hat{\va}_i|^2 \ge P_i/C$. What we do next is address how we
can guarantee such a $C$. We shall move from simpler to more complex
situations, and accordingly start off with i.i.d.\ scenes where for infinitely
many $n$ one does not need the full generality of Nazarov's solution.

\subsubsection{i.i.d.\ scenes}

Recalling that $d(x)=\theta$ is constant, the waterfilling asks
for a $0,1$ sequence with uniform spectrum allocation after the DC term (``spectrally flat
sequences''). As already noted in~\cite{fenimore1978,yedidia2018},
one can certainly construct such spectrally flat sequences for infinitely many
values of $n$, as long as they are ``unbiased'' with $\rho=1/2-o(1)$. This meets the lower bound (as $n
\rightarrow \infty$) as long as the optimal $\rho^*$ is $1/2-o(1)$ for the given $t,W,J$.
A natural question is how good is using an ``unbiased'' spectrally flat
sequence when $\rho^* \ne 1/2$? The answer is given in the following:
\begin{propn}
\label{propn:spec_flat_simple}
Let $\theta,W,J$ be fixed and let $\vd = (\theta/n)\vone$.
Then for infinitely many $n$, there exists a $\va \in [0,1]$ such that:
\begin{equation}
\label{eqn:spec_flat_simple}
m(n,2t,W,J,\vd,\va) \le m^*(n,t,W,J,\vd).
\end{equation}
\end{propn}

In other words, ``unbiased'' spectrally flat sequences are always guaranteed
to achieve optimal LMMSE at the expense of increasing the exposure time $t$ by
a factor of at most $2$. In the sequel, we show how one can reduce this
factor even further.

This is achieved by using spectrally flat sequences with $\rho=1/8-o(1)$ and
$\rho=1/4-o(1)$, and allows us to refine $2$ to $8/7$. The construction of these is
based on well-established cyclotomic number computations~\cite[Art. 356]{gauss1801} in number theory:
$1/4$ corresponds to quartic residues~\cite{chowla1944}, and
$1/8$ corresponds to octic residues~\cite{lehmer1953}. It should be emphasized,
however, that in contrast to the case in which $\rho=1/2$, the
existence of such sequences for infinitely many values of $n$ is not
guaranteed, because no single-variable quadratic taking on infinitely
many primes is known~\cite{mo2018}.

Of perhaps greater importance is the fact that the octic residue constructions
of~\cite{lehmer1953} rely upon primes that come from a second order linear
recurrence with rather large coefficients, arising as the solutions of
Brahmagupta-Pell equations. There is thus a paucity of such constructions,
indeed~\cite{lehmer1953} gives only two such $n$ below $10^9$, namely
$n=73$ and $n=26041$. On the other hand, the quartic residue constructions are
reasonably numerous, with over $150$ of them available below $10^7$. Even
restricting ourselves to the quartic residues allows us to tighten from $2$ to
$4/3$. Summarizing all of the above, we have:
\begin{propn}
\label{propn:spec_flat_comp}
Let $\theta,W,J$ be fixed and let $\vd = (\theta/n)\vone$.
Then for some values of $n$ that exist even beyond, e.g., $10^9$, there exists a
$\va \in [0,1]$ such that:
\begin{subequations}
\label{eqn:spec_flat_comp}
\begin{equation}
\label{eqn:spec_flat_comp_v1}
m(n,(8/7)t,W,J,\vd,\va) \le m^*(n,t,W,J,\vd).
\end{equation}
Moreover, for many ($> 150$ for $n < 10^7$) values of $n$ that exist even
beyond, e.g., $10^9$, there exists a $\va \in [0,1]$ such that:
\begin{equation}
\label{eqn:spec_flat_comp_v2}
m(n,(4/3)t,W,J,\vd,\va) \le m^*(n,t,W,J,\vd).
\end{equation}
\end{subequations}
\end{propn}

\begin{proof}[Proof of Props.~\ref{propn:spec_flat_simple},~\ref{propn:spec_flat_comp}]
The ``difference sets'' of~\cite{chowla1944,lehmer1953} are in our language spectrally flat sequences.
The constant factor is given by the following single variable optimization.
In view of~\eqref{eqn:lmmse_lb}, let $f_a(\rho) = (\rho(1-\rho))/(a+\rho)$ defined on $[0,1]$; $a$ corresponds
to $W/J$. The numerator comes from the power bound, the denominator from the noise penalty.
Then, $M(a,\rho)=(\sup_x f_a(x))/(f_a(\rho))$ is the multiplicative loss factor for a fixed
$W/J$ and fixed $\rho \in \{0.125,0.25,0.5\}$. One may then optimize over
$\rho,a$ to get the constant~\eqref{eqn:spec_flat_comp_v1}. This proof, modified to $\rho \in \{0.25,0.5\}$
and $\rho \in \{0.5\}$, also yields~\eqref{eqn:spec_flat_comp_v2}
and~\eqref{eqn:spec_flat_simple} respectively. The fact that there are
infinitely many $n$ for $\rho=0.5-o(1)$ follows from the quadratic residue
construction together with the well known fact that there are
infinitely many primes $p=4k+3$ (see, e.g., \cite[Chap 7]{apostol2013}).
\end{proof}

\subsubsection{Correlated scenes}

We now turn to correlated scenes. Here the waterfilling is nontrivial, and
asks for an unequal spectrum allocation. We therefore invoke Nazarov's
solution to the coefficient problem~\cite[p. 5]{nazarov1997}, and
provide a statement here specialized to the DFT and $l_\infty$ that we use.

First, some notation. Let us define inner products with respect to the uniform probability
distribution on $\{0,1,\dots,n-1\}$. Let $0 \le i,j \le n-1$, and let $\psi_j$ be a orthonormal basis for the DFT
on real sequences. Explicitly, let $h=\ceil{(n-1)/2}$. Let $\psi_0(i) = 1$,
$\psi_j(i) = \sqrt{2}\cos(\omega ji)$ for $0 < j < h$,
$\psi_j(i) = \sqrt{2}\sin(\omega ji)$ for $h < j < n$.
If $n$ is even, let $\psi_h(i) = \cos(\omega hi)$, otherwise $\psi_h(i) =
\sqrt{2}\cos(\omega hi)$. Finally, let $\beta(n) = \min_{j} |\psi_j|_1$.
\begin{thm}[Nazarov]
\label{thm:nazarov}
Let $M(n) = ((3\pi)/2)\beta(n)^{-2}$. Let $0 \le p_0,p_1,\dots,p_{n-1}$
be such that $\sum p_j = 1$. Then there exists a $\vb \in [-M(n),M(n)]$ with
$|(\vb,\psi_j)|^2 \ge p_j$ for all $0 \le j \le n-1$.
\end{thm}

With~\ref{thm:nazarov} in hand, we are able to reach a far more general version of
Prop.~\ref{propn:spec_flat_simple},~\ref{propn:spec_flat_comp} valid
for any $n$ and any scene prior $\vd$. Also, in Sec.~\ref{sec:greedy} we show
how to actually construct such tight sequences whose existence is guaranteed
by~\ref{thm:nazarov}.
\begin{propn}
\label{propn:nazarov_power}
For all $n,t,W,J,\vd$, there exists a $\va \in [0,1]$ such that:
\begin{equation}
\label{eqn:nazarov_power}
m(n,2M(n)^2t,W,J,\vd,\va) \le m^*(n,t,W,J,\vd).
\end{equation}
Furthermore, we have:
\begin{equation}
\label{eqn:dft_l1_est}
M(n) \in [(3\pi^3)/16+o(1), 3\pi+o(1)].
\end{equation}
\end{propn}

The justification of the tightness of~\eqref{eqn:nazarov_power} lies in
establishing~\eqref{eqn:dft_l1_est}, which we do first.
The phenomenon is captured by the factorization of $n$, with the best, that is
the largest, $\beta$ occurring for $n$ prime, and the worst occuring for
$n$ divisible by $4$. We have the following Lemma which establishes~\eqref{eqn:dft_l1_est}:
\begin{lemma}
\label{lem:beta_bound}
$\beta(n) \in
\left[\frac{1}{\sqrt{2}}+o(1),\frac{2\sqrt{2}}{\pi}+o(1)\right]$ as
$n \rightarrow \infty$. Moreover, if we restrict to $n$ being prime, $\beta(n)
= \frac{2\sqrt{2}}{\pi}+o(1)$.
\end{lemma}

\begin{proof}[Proof sketch]
We give a full proof for the $n=p$ prime case. Then, for
any $j \ne 0$, $ij$ sweeps over $\{0,1,\dots,p-1\}$, modulo $p$. Thus, really one is
looking at a Riemann sum approximation to $\int_0^1|\cos(2\pi x)| dx = 2/\pi$. The $l_2$ norm of
$\cos(2\pi x)$ on $[0,1]$ is $1/\sqrt{2}$, completing the prime case.
The composite case is more involved, as it needs to take into account the divisor
structure of $n$, which prevents such symmetry of the cosine vectors. Once accounted
for, the natural idea is to use Euler-Maclaurin summation, with standard
modifications by, e.g., mollifiers
to take into account the lack of smoothness of $|\cos(x)|$ at its zeros.
However, the mechanics are perhaps simplest in our specific setting when one uses short quadratic splines around
the zeros to get a $C^1$ approximation of any desired accuracy to $|\cos(x)|$ while not changing the uniform derivative bound.
We omit a full proof due to space constraints; see, e.g., \cite{tao2018} for the
mechanics of how this is done in general.
\end{proof}

We emphasize that by Lemma~\ref{lem:beta_bound} $M(n) \le C$ for some
universal constant $C \approx 9.4248$, with even better values available at,
e.g., prime $n > 100$. There, $C \approx 5.8146$ suffices.
\begin{proof}[Proof of Prop.~\ref{propn:nazarov_power}]
Thm.~\ref{thm:nazarov}, with $p_0 = 0$ and $p_j = P_j/\sum_j P_j$ for $0 <
j < n$ yields a $\vb$ with $|\vb|_\infty \le M(n)$ and $|(\vb,\psi_j)|^2 \ge p_j$
for $0 < j < n$. Without loss, we may assume that $(\vb,\psi_0) \le 0$, else
simply flip signs. Stitching the $\psi_j$ back to complex exponentials and
recalling the upper bound $P \le n^2\rho(1-\rho)$, this
gives $|\hat{\vb}_j|^2 \ge P_j/(\rho(1-\rho))$. Consider $\va = (\vb+M(n))/2M(n)$.
Then, $\va \in [0,1]$, $\rho(\va) \le 0.5$, and $|\hat{\va}_j|^2 \ge P_j/(4M(n)^2\rho(1-\rho))$
for $0 < j < n$. We are now in a similar situation to that of
Prop.~\ref{propn:spec_flat_simple}, except with an extra $M(n)^2$
factor, and the fact that $\rho(\va) \le 0.5$ instead of $\rho(\va) =
0.5+o(1)$. The latter is no problem, as lower $\rho$ only helps us with the
shot noise term, and the former simply multiplies the 2 of~\eqref{eqn:spec_flat_simple} by $M(n)^2$.
\end{proof}

\subsection{Greedy algorithm}
\label{sec:greedy}
Here we propose a (heuristically) efficient algorithm to construct vectors
$\va$ that satisfy the conditions of Prop.~\ref{propn:nazarov_power}. This
algorithm has its roots in Nazarov's original proof. At a high level,
Nazarov's theoretical construction boils down to finding a ``sign
cort\`{e}ge''~\cite[p. 6]{nazarov1997} that is globally optimal for a certain
real-valued Boolean function of $n$ signs, taking exponential time in the worst case.
However, a closer examination of Nazarov's proof reveals that one simply needs
a sign cort\`{e}ge that is locally optimal in the sense of Hamming geometry for
the proof to work. Our observation suggests a natural greedy algorithm where one
starts with a random cort\`{e}ge, and then flips one sign at a time if it improves the objective,
repeating until no further improvement is possible. In our simulations
\footnote{Code:\url{https://github.com/gajjanag/apertures}}
this runs very fast. For example, on our standard laptop, we can generate
apertures for $n=2000$ in $4$ seconds. This superficially resembles the situation of the
simplex algorithm and the smoothed analysis of~\cite{spielman2001},
or more directly recent work on max-cut~\cite{angel2017}. Direct application of the
methods of~\cite{angel2017} to obtain theoretical guarantees runs into
difficulties with the nonlinear change in objective with a single bit flip in
our setting, unlike the linear change for max-cut. As such, we defer
theoretical study of the greedy algorithm given here to future work.

\subsection{Simulations}
We give a simple illustration in Fig.~\ref{fig:mmse_plot} which confirms the
following intuition based on our main results~\cref{eqn:spec_flat_simple,eqn:spec_flat_comp,eqn:nazarov_power}.
With an i.i.d.\ scene prior, one would prefer using the spectrally flat construction
as opposed to the one coming from Nazarov's theorem due to the smaller
constant. On the other hand, with a strong prior---e.g., a bandlimited one---the
waterfilling becomes highly skewed, and one would favor the one coming
from Nazarov's theorem as it takes into account such strong skewing of the
desired spectrum. For completeness, we also include the performance of a
random on-off sequence with density $\rho$~\cite{yedidia2018}, where $\rho$ is
optimized over $[0,1]$ for each $t$.
\begin{figure}[t]
\begin{minipage}[b]{1.0\linewidth}
\centering
\centerline{\includegraphics[width=1.0\linewidth]{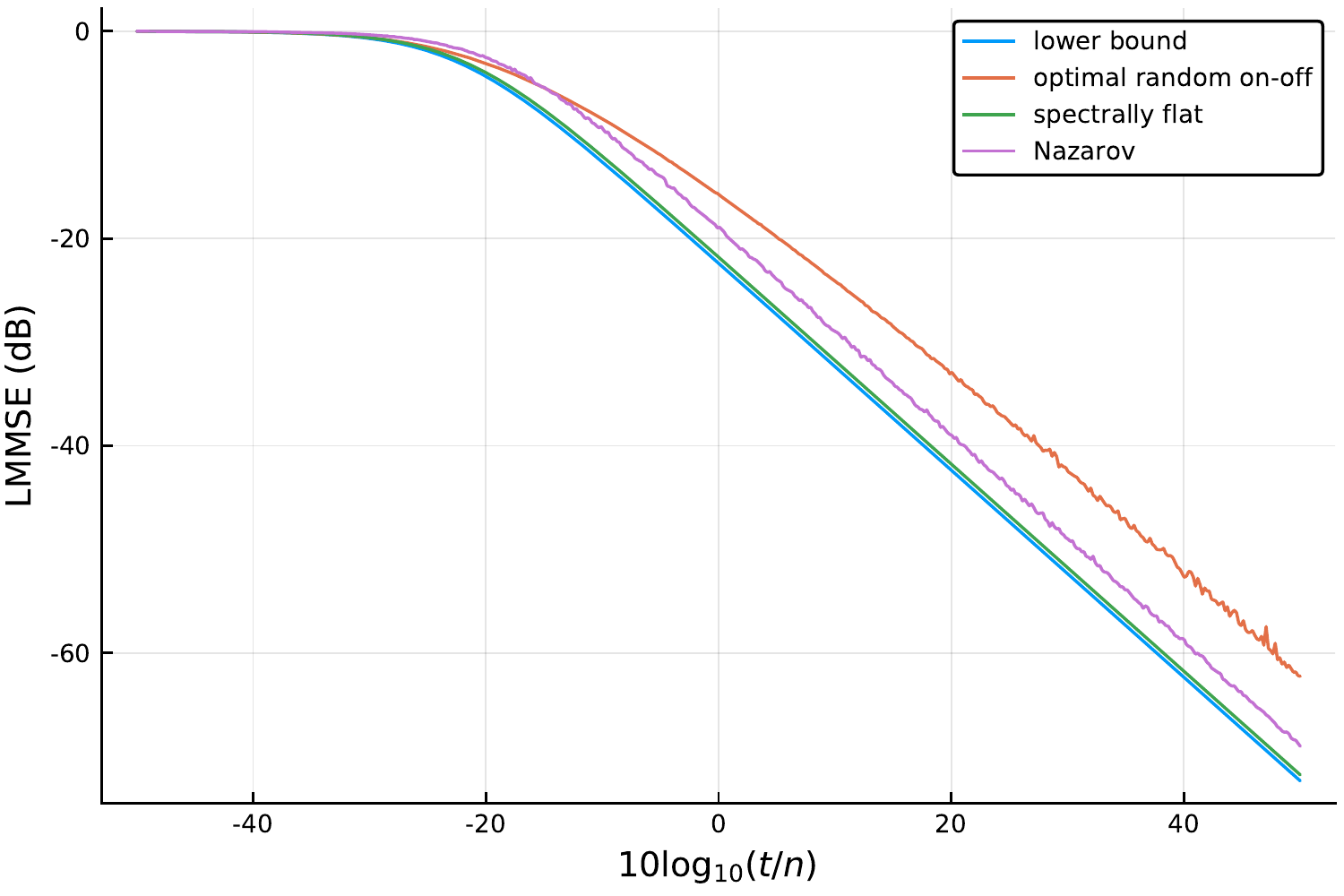}}
\centerline{(a) i.i.d prior ($d(x)=\theta$ for $0 \le x \le 1/2$)}
\end{minipage}
\begin{minipage}[b]{1.0\linewidth}
\centering
\centerline{\includegraphics[width=1.0\linewidth]{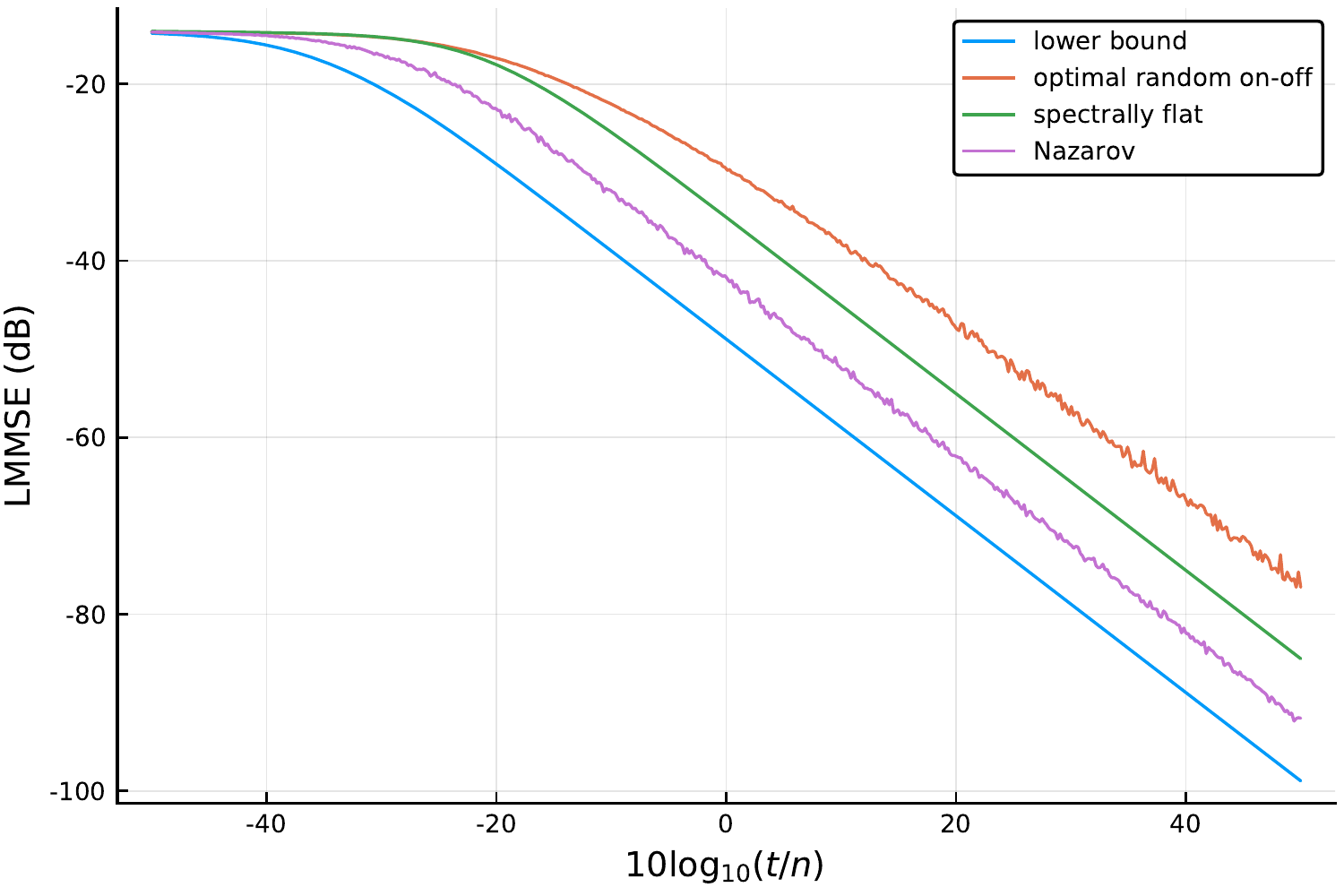}}
\centering{(b) bandlimited prior ($d(x) = \theta$ for $0 \le x \le s-r$,
$0$ for $x \ge s+r$, and $\theta(s+r-x)/(2r)$ otherwise for $0 \le x \le 1/2$)}
\end{minipage}
\caption{$n=677,\theta=1,W=J=0.001, s=0.02, r=0.005$. We use the quartic residue construction for
spectrally flat. Jaggedness of the Nazarov plot comes from the fact that in
general the spectrum allocation varies with $t$ and we randomly seed the sign cort\`{e}ge.}
\label{fig:mmse_plot}
\end{figure}
\section{Discussion and Future Work}
\label{sec:conclusion}
Our refined analysis of a model drawing heavily from~\cite{yedidia2018} yields
tight conclusions across all scene correlation patterns and noise regimes, with sharp
conclusions available in some specific scenarios. Moreover, we give
heuristically efficient algorithms for the generation of optimal coded
apertures. We also note that similar conclusions to our main
results~\cref{eqn:spec_flat_simple,eqn:spec_flat_comp,eqn:nazarov_power}
also hold for MI and Gaussian statistics of~\cite{yedidia2018},
simply because of the form of the expression for MI.

Furthermore, we note that our conclusions generalize naturally to 2D
apertures, and in particular we have a tight
characterization of optimal coded apertures in that setting. Concretely, one
simply needs to take $\beta(n)^2$ as opposed to $\beta(n)$ due to the squaring
of the $l_1$ lower bound for the 2D DFT. The rest of the analysis of
Thm.~\ref{thm:nazarov} and Prop.~\ref{propn:nazarov_power} carries over
naturally, with the orthogonal basis provided by products of $\psi_j$.
We emphasize that this works regardless of the scene prior, even
ones which are not separable. With an i.i.d.\ prior, separable apertures are
optimal up to constants as in 1D, and in fact taking a product of spectrally flat apertures yields
natural analogs of Props.~\ref{propn:spec_flat_simple},~\ref{propn:spec_flat_comp}.
However, with other priors, it seems like one needs the generality provided by
Thm.~\ref{thm:nazarov}. This work thus also answers the question of 2D
apertures raised in~\cite{yedidia2018}. We also view experimental verification
of these ideas as a worthwhile task.

As noted in~\cite{yedidia2018},~\cite{veeraraghavan2007} raises the
question of whether continuous-valued masks perform better than binary-valued
ones. This work sheds some light on this: the solution of Nazarov which we have
shown is tight does seem to use the flexibility of the $l_\infty$ norm in an
essential way; see, e.g., \cite[p. 12]{green2004} for more on this. And more
specifically, we have numerical evidence for finite $n$; to give a
concrete example, for $n=13$, the mask $[1, 0, 1, 0, 0, 1, 1, 0, 1, 1, 0, 0,
0]$ has optimal LMMSE for an i.i.d.\ scene over binary-valued masks for
$\rho=6/13, \theta=0.01, W=J=0.001, t=130$, but is improved upon by the
continuous-valued mask whose first entry is equal to $\epsilon$ and whose
$i$\/th entry is equal to $1-\epsilon/6$ if $i-1$ is a quadratic
residue modulo 13, and $0$ otherwise, for $0.26 \le \epsilon \le 0.34$.

Although Prop.~\ref{propn:nazarov_power} shows universal tightness across all
priors, even ``extreme'' ones like bandlimited ones, the constant is worse
than that for a spectrally flat construction for i.i.d.\ scenes. The better
performance of spectrally flat constructions over the ones inspired by
Nazarov's theorem seems to extend to other ``natural''
priors like the $f^{-\gamma}$ one, as the waterfilling still yields something
that is nearly ``flat''. It might be interesting to quantify and understand
the ``flatness'' of the waterfilling for ``natural'' priors.

One issue that we have not addressed here or in~\cite{yedidia2018} is the equal scaling
of $n$ at both sensor and scene. One natural way to address this is letting
$\mA$ be $m \times n$, or alternatively one could study a continuous model.
Another issue is obtaining a good understanding of mask/lens combinations. This will require not only updates
to the simple propagation model studied here and in~\cite{yedidia2018},
but also a refined understanding of the cost tradeoffs between lenses and apertures.

Stepping back from imaging problems, one may ask the question of where
else Nazarov's theorem can be used in applied contexts, something also
raised implicitly in~\cite{boche2016}. For example, as Nazarov's theorem does not care
about orthogonality, but merely a $l_2$ estimate like Parseval's theorem, one
can use it for frames as well as bases, or for anything satisfying a
restricted isometry property. Another example is the fact that we
merely use the $l_\infty$ case of his theorem which works for all $l_p$
spaces.

\section{Acknowledgement}
Ganesh Ajjanagadde thanks Prof. Henry Cohn for discussions on pseudorandomness.

\vfill\pagebreak

\bibliographystyle{IEEEbib}
\bibliography{mask}

\end{document}